\documentclass[11pt]{article}
\usepackage{amssymb,amsmath,amsthm,amscd,latexsym}
\usepackage{mathrsfs}
\usepackage{mathrsfs}
\usepackage{amsfonts}
\usepackage{amsmath}
\usepackage{amssymb}
\usepackage{amscd}

\renewcommand{\paragraph}{\roman{paragraph}}
\input cyracc.def

\newtheorem{thm}{\scshape \mdseries  Theorem}[section]
\newtheorem{lem}[thm]{\scshape \mdseries  Lemma}

\newtheorem{rem}[thm]{\scshape  \mdseries Remark}

\begin{document}
%\begin{CJK*}{GBK}{song}\CJKtilde
\title{\bf On self-dual four circulant codes
\thanks{This research is supported by National Natural Science Foundation of China (61672036), Technology Foundation for Selected Overseas Chinese Scholar, Ministry of Personnel of China (05015133) and the Open Research Fund of National Mobile Communications Research Laboratory, Southeast University (2015D11) and Key projects of support program for outstanding young talents in Colleges and Universities (gxyqZD2016008).}
}
\author{
\small{Minjia Shi}\\
\small { Key Laboratory of Intelligent Computing \& Signal Processing, Ministry of Education,}\\
\small { Anhui University No. 3 Feixi Road, Hefei 230039, China,}\\
\small{National Mobile Communications Research Laboratory,}\\
\small{ Southeast University, Nanjing 210096, China,}\\
 \small{School of Mathematical Sciences of Anhui University, Hefei 230601, China }\\
\small{Hongwei Zhu}\\\small{School of Mathematical Sciences of Anhui University, Hefei, 230601, China }
\and \small{Patrick Sol\'e}\\ \small{CNRS/LAGA, University of Paris 8, 2 rue de la libert\'e, 93 526 Saint-Denis, France}
}

\date{}
\maketitle
%\end{CJK*}
%\begin{CJK}{GBK}{song}
\begin{abstract}
  Four circulant codes form a special class of $2$-generator, index $4$, quasi-cyclic codes. Under some conditions on their generator matrices they can be shown to be self-dual. Artin primitive root conjecture shows the existence of an infinite subclass of these codes satisfying a modified Gilbert-Varshamov bound.
\end{abstract}

\section{Introduction}
A matrix over a finite field $\mathbb{F}_q$ is said to be {\it circulant} if its rows are obtained by successive shifts from the first row.
If two circulant matrices $A$ and $B$ of order $n$ satisfy $AA^T+BB^T+I_n=0$, where the exponent $``T"$ denotes transposition, then the code $C$ generated by
\begin{equation}\label{den1}
 G=\left(
  \begin{array}{cccc}
    I_n & 0 & A & B \\
    0 & I_n & -B^T & A^T \\
  \end{array}
\right)
\end{equation}
is a self-dual code. It is a $[4n,2n]$ code. This so-called {\it four circulant construction} was introduced in \cite{KS}, revisited in \cite{K2}, and  many self-dual codes with good parameters have been constructed that way.

In the above generator matrix $G$, $A$ and $B$ are circulant matrices and are determined by the polynomials $a(x), b(x)\in\mathbb{F}_q[x]$ whose $x$-expansion are the first row of the matrices $A$ and $B$, respectively. A code of length $4n$, dimension $2n$ is a four circulant code if it is generated by $G$.

A linear code of length $N$ over $\mathbb{F}_q$ is a {\it quasi-cyclic} code of index $l$ for $l\mid N$ if for each codeword $\mathbf{c}=(c_0,c_1,\ldots,c_{N-1})$ in $C$, the vector $(c_{N-l},c_{N-l+1},\ldots,c_{N-1},c_0,\ldots,c_{N-l-1})\in C$, where the subscripts are taken modulo $N$.
Hence four circulant codes are quasi-cyclic codes of index $4$. Quasi-cyclic codes form an important class of codes, which have been extensively studied \cite{J,L1,L2}. By the Chinese Remainder Theorem(CRT), quasi-cyclic codes can be decomposed into a ``CRT product" of shorter codes over larger alphabets \cite{L1}. Such a linear code affords a natural structure of module over the auxiliary ring $R(n,\mathbb{F}_q)=\mathbb{F}_q[x]/\langle x^n-1\rangle.$ A linear code $C$ of length $4$ over $R(n,\mathbb{F}_q)$ is an $R(n,\mathbb{F}_q)$-submodule of $R^4(n,\mathbb{F}_q)$. The ring $R(n,\mathbb{F}_q)$ is never a finite field except if $n=1$. The CRT shows us that, if $n$ is coprime with the characteristic of $\mathbb{F}_q$, then the ring is a direct product of finite fields.  If this $R(n,\mathbb{F}_q)$-submodule is generated by $r$ generator sets, the code will be called $r$-generator code. In view of the generator matrix structures of four circulant codes, we see that these codes are $2$-generator.

In the present paper, we study index $4$ quasi-cyclic codes of parameters $[4n,2n]$ inspired by  \cite{A1,A3,ML}. When $x^n-1$ has only two  irreducible factors, we derive an enumeration formula of the self-dual four circulant codes over $\mathbb{F}_q$ based on that decomposition, and derive an asymptotic lower bound on the minimum distance of these four circulant codes. The asymptotic lower bound is in the spirit of the Gilbert-Varshamov bound. Previous articles \cite{C1,K1} also explore this expurgated random coding technique for other families of structured codes.

The material is organized as follows. Section 2 collects the necessary notions and definitions. Section 3 discusses the algebraic structure of four circulant codes. Section 4 presents enumeration formula when $x^n-1$ only has two irreducible factors over $\mathbb{F}_q$ and establishes the asymptotic of self-dual four circulant codes. Section 5 concludes this article.

\section{\textbf{Definitions and notations}}
Let $\mathbb{F}_q$ denote a finite field of characteristic $p,$ and $R(n,\mathbb{F}_q)=\mathbb{F}_q[x]/\langle x^n-1\rangle$. In the following, we consider codes over $\mathbb{F}_q$ of length $4n$ with $n$ coprime to $q$. The code $C$ is generated by a matrix of the form (1).

 From an algebraic perspective, we can view such a code $C$ as an $R(n,\mathbb{F}_q)$-submodule in $R^4(n,\mathbb{F}_q)$ and its two generators are $((1,0,a(x),b(x)),(0,1,-b^\prime(x),a^\prime(x)))$, where $a^\prime(x)=a(x^{n-1})$~\rm{mod}~$(x^n-1)$, $b^\prime(x)=b(x^{n-1})$~\rm{mod}~$(x^n-1)$.

If $C(n)$ is a family of codes parameters $[n,k_n,d_n]$, the rate $\alpha_q(\delta)$ and relative distance $\delta$ are defined as
$$\alpha_q(\delta)=\limsup\limits_{n\rightarrow \infty}\frac{k_n}{n}$$
and
$$\delta=\liminf\limits_{n\rightarrow \infty}\frac{d_n}{n}.$$
Both limits are finite as limits of bounded quantities.
\section{\textbf{ Algebraic structure of self-dual four circulant codes}}
Throughout this paper, we assume $\gcd(n,q)=1$. According to \cite{L1}, we can cast the factorization of $x^n-1$ into distinct irreducible polynomials over $\mathbb{F}_q$ in the form
$$x^n-1=\alpha\prod\limits_{i=1}^sg_i(x)\prod\limits_{j=1}^th_j(x)h_j^*(x),$$
where $\alpha\in \mathbb{F}_q^*,$ $g_i(x)$ is a self-reciprocal polynomial for $1\leq i\leq s$, and $h_j^*(x)$ is the reciprocal polynomial of $h_j(x)$ for $1\leq j\leq t$.

By the Chinese Remainder Theorem (CRT), we have
\begin{eqnarray*}
  R &\simeq& \left(\bigoplus_{i=1}^s\mathbb{F}_q[x]/\langle g_i(x)\rangle\right)\oplus\left(\bigoplus_{j=1}^t(\mathbb{F}_q[x]/\langle h_j(x)\rangle\oplus(\mathbb{F}_q[x]/\langle h_j^*(x)\rangle))\right). \\
\end{eqnarray*}
Let $G_i=\mathbb{F}_q[x]/\langle g_i(x)\rangle, H_j^\prime=\mathbb{F}_q[x]/\langle h_j(x)\rangle, H_j^{\prime\prime}=\mathbb{F}_q[x]/\langle h_j^*(x)\rangle$ for simplicity. Note that all of these fields are extensions of $\mathbb{F}_q$. This decomposition naturally extends to $R^4$ as
$$R^4\simeq\left(\bigoplus_{i=1}^s{G_i}^{4}\right)\oplus\left(\bigoplus_{j=1}^t({H_j^\prime}^{4}\oplus{H_j^{\prime\prime}}^4)\right).$$
In particular, each $R(n,\mathbb{F}_q)$-linear code $C$ of length $4$ can be decomposed as the ``CRT sum"
$$C\simeq\left(\bigoplus_{i=1}^s{C_i}\right)\oplus\left(\bigoplus_{j=1}^t({C_j^\prime}\oplus{C_j^{\prime\prime}})\right),$$
where for each $1\leq i\leq s$, $C_i$ is a linear code over $G_i$ of length $4$, and for each $1\leq j\leq t$, $C_j^\prime$ is a linear code over $H_j^\prime$ of length $4$ and $C_j^{\prime\prime}$ is a linear code over $H_j^{\prime\prime}$ of length $4$, which are called the constituents of $C$.

 Note that in terms of constituents of $C$, we have $C$ is self-dual only if $C_i$ is self-dual relative to Hermitian product in $G_i^4$ for $1\leq i\leq s$, and $C_j^{\prime}\bigcap {C_j^{\prime\prime}}^{\bot}\neq\{0\}$, $C_j^{\prime\prime}\bigcap {C_j^{\prime}}^{\bot}\neq\{0\}$ for $1\leq j\leq t.$ Similar to the proof of \cite{ML}, we can get the following theorem.
 \begin{thm}\label{a1}
 If $C=((1,0,a(x),b(x)),(0,1,-b^\prime(x),a^\prime(x)))\subset R^4$ be a four circulant code over $\mathbb{F}_q$, then $C$ is self-dual if and only if $(x^n-1)\mid (1+a(x)a(x^{n-1})+b(x)b(x^{n-1}))$.
\end{thm}
\begin{rem}
Under the condition of Theorem \ref{a1}, then $C$ is a linear complementary-dual four circulant code if and only if $gcd(1+a(x)a(x^{n-1})+b(x)b(x^{n-1}),x^n-1)=1$.
\end{rem}

\section{\textbf{ Asymptotic of self-dual four circulant codes}}
\subsection{\textbf{Enumeration}}
In this subsection, we will give enumerative results for self-dual four circulant codes. Recall that the so-called quadratic character $\eta$ of $\mathbb{F}_q$ is defined as $\eta(x)=1$ if $x\in \mathbb{F}^*_q$ is square and $\eta(x)=-1$ if not.

For our purpose, we first give some lemmas without proofs as follows.
\begin{lem} (\cite[Appendix]{A2})\label{a3}
If $q$ is odd, then the number of solutions $(x,y)$ in $\mathbb{F}_{q}$ of the equation $x^2+y^2=-1$ is $q-\eta(-1)$.
\end{lem}

The proof of the next Lemma while given in odd characteristic in \cite{A2} is easily seen to hold more generally, for all prime powers $q.$
\begin{lem}\label{a2} (\cite[Appendix]{A2})
If $n$ is coprime with $q$, then the number of solutions $(a,b)$ in $\mathbb{F}_{q^2}$ of the equation $a^{1+q}+b^{1+q}=-1$ is $(q+1)(q^2-q)$.
\end{lem}

On the basis of Lemmas 4.1 and 4.2, we are now ready to state and prove the main result of this subsection.

\begin{thm}\label{a5}
Let $n$ be an odd prime, $\gcd(n,q)=1$ and $q$ be a primitive root modulo $n$. Then, $x^n-1=(x-1)h(x)$ can be factored as a product of two irreducible polynomials over $\mathbb{F}_q$ and the number of self-dual four circulant codes of length $4n$ is $(q-\eta(-1))(q^{\frac{n-1}{2}}+1)(q^{n-1}-q^{\frac{n-1}{2}})$ if $q$ is odd and $q(q^{\frac{n-1}{2}}+1)(q^{n-1}-q^{\frac{n-1}{2}})$ if $q$ is even.
\end{thm}
\begin{proof}
Since $q$ is a primitive root modulo $n$, the cyclotomic cosets of $q$ modulo $n$ are only two in number, that is $C_0=\{0\}$ and $C_1=\{1,q,\ldots,q^{n-2}\}$. Due to $\gcd(n,q)=1$, the number of monic irreducible factors of $x^n-1$ over $\mathbb{F}_q$ is equal to the number of cyclotomic cosets of $q$ modulo $n$. Hence $x^n-1$ can be factored as a product of two irreducible polynomials over $\mathbb{F}_q$. Let $x^n-1=(x-1)h(x)$, where $h(x)$ is a irreducible polynomial and  $\deg(h(x))=n-1$.

By the CRT approach of \cite{L1}, $C$ can be decomposed as $C\simeq C_1\oplus C_2$, where $C_1\subset \mathbb{F}_q^4$ and $C_2\subset (\mathbb{F}_q[x]/\langle h(x)\rangle)^4$. In the case $C_1\subset \mathbb{F}_q^4$, we count self-dual four circulant codes of parameters $[4,2]$ over $\mathbb{F}_q.$ We can obtain the equation
$1+a^2+b^2=0.$
When $q$ is odd, it follows from Lemma 4.1 that the number of the solutions of that equation is $q-\eta(-1)$. When $q$ is even, then the equation becomes $a^2+b^2=(a+b)^2=1.$ Hence the number of the solutions of that equation is $q$.

In the case $C_2\subset (\mathbb{F}_q[x]/\langle h(x)\rangle)^4$, the factor $h(x)$ is a self-reciprocal polynomial of degree $n-1$, and the number of self-dual four circulant codes of parameters $[4,2]$ over $\mathbb{F}_q[x]/\langle h(x)\rangle$ is equal to the number of solutions of the equation
$1+aa^{q^{\frac{n-1}{2}}}+bb^{q^{\frac{n-1}{2}}}=0$. It follows from Lemma 4.2 that the number of the solutions of that equation is $(q^{\frac{n-1}{2}}+1)(q^{n-1}-q^{\frac{n-1}{2}}).$ Hence the number of self-dual four circulant codes of length $4n$ is $(q-\eta(-1))(q^{\frac{n-1}{2}}+1)(q^{n-1}-q^{\frac{n-1}{2}})$ if $q$ is odd and $q(q^{\frac{n-1}{2}}+1)(q^{n-1}-q^{\frac{n-1}{2}})$ if $q$ is even.
\end{proof}

\subsection{\textbf{Distance bounds}}
 In number theory, Artin's conjecture on primitive roots states that a given integer $q$ which is neither a perfect square nor $-1$ is a primitive root modulo infinitely many primes $l$ \cite{P}. It was proved conditionally under $GRH$ by Hooley \cite{C}.
 In this subsection, we study the case when $x^n-1$  factors as a product of two irreducible polynomials over $\mathbb{F}_q$, i.e. $x^n-1=(x-1)h(x)$, where $h(x)$ is an irreducible polynomial over $\mathbb{F}_q$. We call {\it constant vectors} if the codewords of the cyclic code of length $n$ generated by $h(x)$. For example, when $n=7$, $q=3$, then $x^7-1=(x-1)(x^6+x^5+x^4+x^3+x^2+x+1)$.
Let $a(x)$ and $b(x)$ denote a polynomial of $\mathbb{F}_q[x]$ coprime with $x^n-1$, and let $C_{a,b}$ be the self-dual four circulant code with generator $((1,0,a(x),b(x)),(0,1,-b^\prime(x),a^\prime(x)))$.
\begin{lem}\label{a4}
If $u=(c,d,e,f)$, where $c, d$ is not a constant vector, then there are at most $q^n(q-1)$ generators $((1,0,a(x),b(x)),(0,1,-b^\prime(x),a^\prime(x)))$ such that $u\in C_{a,b}.$
\end{lem}
\begin{proof}
The condition $u=(c,d,e,f)$ is equivalent to the system of equations
\begin{equation*}\label{den1}
\begin{cases}
 \emph{ }e\equiv ca-db^\prime\mod (x^n-1),  \\
   \emph{ }f^\prime\equiv c^\prime b^\prime+d^\prime a\mod (x^n-1). \\
\end{cases}
\end{equation*}
If $cc^\prime+dd^\prime$ is invertible~\rm{mod}~$(x^n-1)$, then
\begin{equation}\label{den1}
\begin{cases}
 \emph{ }a\equiv \frac{ec^\prime+df^\prime}{cc^\prime+dd^\prime}\mod (x^n-1),  \\
   \emph{ }b\equiv \frac{fc^\prime-de}{cc^\prime+dd^\prime}\mod (x^n-1). \\
\end{cases}
\end{equation}
In that case $a$ and $b$ are uniquely determined by this system of equations.
If $cc^\prime+dd^\prime$ is zero or a zero divisor~\rm{mod}~$(x^n-1)$, we assume there are solutions to the system. So $e\equiv ca-db^\prime\mod (x^n-1)$. Since $\deg a(x)\leq n-1$, then there are only at most $q^n$ choices for $a.$ Given $a$, we will have $(q-1)$ choices for $b$ by the CRT since $d$ is not a constant vector. In total, there are at most $q^n(q-1)$ choices for $(a,b).$
\end{proof}

Recall the $q$-ary entropy function defined for $0\leq t\leq\frac{q-1}{q}$ by
\begin{equation*}\label{den1}
H_q(t)=\begin{cases}
 \emph{ }0,  ~~~~~~~~~~~~~~~~~~~~~~~~~~~~~~~~~~~~~~~~~~~~~~~~~~~~~~~~~~~~~~~~if~~ t=0,\\
   \emph{ }t{\rm{log}}_q(q-1)-t{\rm{log}}_q(t)-(1-t){\rm{log}}_q(1-t), ~~~~~~~~~~~if~~0<t\leq\frac{q-1}{q}. \\
\end{cases}
\end{equation*}
This quantity is instrumental in the estimation of the volume of high-dimensional Hamming balls when the base field is $\mathbb{F}_q$.
The result we are using is that the volume of the Hamming ball of radius $tn$ is, up to subexponential terms, $q^{nH_q(t)}$, when
$0<t<1$, and $n$ goes to infinity \cite[Lemma 2.10.3]{W}.
The main  result obtained in this paper is as follows.
\begin{thm}\label{lem2}
Let $n$ be odd prime, $\gcd(n,q)=1$ and $q$ be a primitive root modulo $n$, then there are infinite families of self-dual four circulant codes of relative distance $\delta$ satisfying $H_q(\delta)\geq \frac{1}{8}$.
\end{thm}
\begin{proof}
The four circulant codes containing a vector of weight $d\sim 4\delta n$ or less are by standard entropic estimates of \cite[Lemma 2.10.3]{W} and Lemma 4.4  bounded above by a quantity of  order at most $q^n(q-1) \times q^{4nH_q(\delta)}=O(q^{n(1+4H_q(\delta))})$, up to subexponential terms. This number will be less than the total number of self-dual four circulant codes, which is, by Theorem 4.3  of the order $O(q^{3n/2}),$ as long as $H_q(\delta)< \frac{1}{8}+\epsilon.$ This ensures the existence of such codes of distance $H_q^{-1}(\frac{1}{8})-\epsilon. $ Letting $\epsilon \rightarrow 0,$ the result follows.
\end{proof}

\section{\textbf{Conclusion and Open problems}}
In this paper, we have considered self-dual four circulant codes. Inspired by \cite{ML}, we have studied four circulant codes from the algebraic perspective and given an enumeration formula of the self-dual subclass in a special case of the factorization of $x^n-1$ . This paper can be considered as a companion paper of \cite{ML}. The main difference between the two papers is the nontrivial nature of the existence of a factorization of $x^n-1$ into two irreducible polynomials, which requires Artin's conjecture on primitive roots, while the same problem for $x^n+1$ can be solved by elementary means. We have only considered enumeration formulae in the case that the factorization of $x^n-1$ consists of  two irreducible polynomials. It would be a worthwhile task to relax this condition by looking at lengths where the factorization of $x^n-1$ into irreducible polynomials contains more than two elements.
In fact extending our enumerative results to a general factorization of $x^n-1$ seems to be a difficult task, leading to solving complex diagonal equations over finite fields.


\begin{thebibliography}{1}
\bibitem{A1} A. Alahmadi, C. G$\ddot{u}$neri, B. $\ddot{O}$zkaya, H. Shoaib, P. Sol$\acute{e}$, On self-dual double negacirculant codes, Discrete Applied Mathematics, 222 (2017) 205-212.
\bibitem{A2} A. Alahmadi, C. G$\ddot{u}$neri, B. $\ddot{O}$zkaya, H. Shoaib, P. Sol$\acute{e}$, On linear complementary-dual multinegacirculant codes,  arXiv:1703.03115v1 [cs.IT], 2017.
\bibitem{A3} A. Alahmadi, F. Ozdemir, P. Sol$\acute{e}$, On self-dual double circulant codes, arXiv:1603.00762v1 [cs.IT], 2016.
\bibitem{C} C. Hooley, On Artin's conjecture, J. Reine Angew. Math, 225 (1967) 209-220.
\bibitem{KS} K. Betsumiya, S. Georgiou, T. A. Gulliver, M. Harada and Koukouvinos, On self-dual codes over some prime fields, Discrete Mathemstics, 262 (2003) 37-58.
\bibitem{C1} V. Chepyzhov, A Varshamov-Gilbert bound for quasi-twisted codes of rate $\frac{1}{n}$, in: Proceedings of the joint Swedish-Russian International workshop on Information Theory, M$\ddot{O}$lle, Sweden, (1993) 214-218.
\bibitem{K2} A. Kaya, B. Yildiz, A. Pasa, New extremal binary self-dual codes from a modified four circulant construction, Discrete Mathematics, 339 (2016) 1086-1094.



%\bibitem{J1} J. Wolfmann, The number of solutions of certain diagonal equations over finite fields, J. of Number Theory, 42 (1992) 247¨C257.

\bibitem{ML} M. Shi, L. Qian, P. Sol$\acute{e}$, On the self-dual four-negacirculant codes, submitted.
\bibitem{P} P. Moree, Artin's primitive root conjecture a survey, Integers, 10(6) (2012) 1305-1416.

\bibitem{W} W. C. Huffman, V. Pless, {\it Fundamentals of error correcting codes}, Cambridge University Press, 2003.
\bibitem{J} Y. Jia, On quasi-twisted codes over finite fields, Finite Fields and their Applications, 18 (2012) 237-257.
\bibitem{K1} T. Kasami, A Gilbert-Varhamov bound for quasi-cyclic codes of rate $\frac{1}{2}$, IEEE Transactions on Information Theory, 20 (1974) 679.

\bibitem{L1} S. Ling, P. Sol$\acute{e}$, On the algebraic structure of quasi-cyclic codes I: finite fields, IEEE Transactions on Information Theory, 47 (2001) 2751-2760.
\bibitem{L2} S. Ling, P. Sol$\acute{e}$, On the algebraic structure of quasi-cyclic codes II: chain rings, IEEE Transactions on Information Theory, 30 (2003) 113-130.
%\bibitem{TN} T. P. Berger, N. E. Amrani, Codes over finite quotients of polynomial rings, Finite Fields and Their Applications, 25 (2014) 165-181.
\end{thebibliography}
\end{document}